\documentclass[11pt]{article}
\usepackage[  textheight=9in, textwidth=6.5in, top=1in, headheight=14pt, headsep=25pt, footskip=30pt]{geometry}
\usepackage[T1]{fontenc}
\usepackage[final]{microtype}
\usepackage[english]{babel}
\usepackage[utf8]{inputenc} 
\usepackage{amsmath,amssymb,amsfonts,mathtools,amsthm}
\usepackage{thmtools}
\usepackage{thm-restate}
\usepackage{bbm}
\usepackage{color,graphicx}
\usepackage{verbatim}
\usepackage{algorithmic}
\usepackage{algorithm}
\usepackage{hyperref}

\newtheorem{theorem}{Theorem}

\newtheorem{definition}[theorem]{Definition}
\newtheorem{problem}[theorem]{Problem}
\newtheorem{claim}[theorem]{Claim}

\DeclareMathOperator{\E}{\mathbb{E}}

\date{}
\usepackage{adjustbox}

\title{Minimizing Cost Rather Than Maximizing Reward \\ in Restless Multi-Armed Bandits}

\author{
    R. Teal Witter \\
    New York University \\
    \texttt{\href{mailto:rtealwitter@nyu.edu}{rtealwitter@nyu.edu}}
    \and
    Lisa Hellerstein \\
    New York University \\
    \texttt{\href{mailto:lisa.hellerstein@nyu.edu}{lisa.hellerstein@nyu.edu}}
}

\begin{document}

\maketitle

\begin{abstract}
Restless Multi-Armed Bandits (RMABs) offer a powerful framework for solving resource constrained maximization problems. However, the formulation can be inappropriate for settings where the limiting constraint is a reward threshold rather than a budget. We introduce a constrained minimization problem for RMABs that balances the goal of achieving a reward threshold while minimizing total cost. We show that even a bi-criteria approximate version of the problem is PSPACE-hard. Motivated by the hardness result, we define a decoupled problem, indexability and a Whittle index for the minimization problem, mirroring the corresponding concepts for the maximization problem. Further, we show that the Whittle index for the minimization problem can easily be computed from the Whittle index for the maximization problem. Consequently, Whittle index results on RMAB instances for the maximization problem give Whittle index results for the minimization problem.  Despite the similarities between the minimization and maximization problems, solving the minimization problem is not as simple as taking direct analogs of the heuristics for the maximization problem. We give an example of an RMAB for which the greedy Whittle index heuristic achieves the optimal solution for the maximization problem, while the analogous heuristic yields the worst possible solution for the minimization problem. In light of this, we present and compare several heuristics for solving the minimization problem on real and synthetic data. Our work suggests the importance of continued investigation into the minimization problem.
\end{abstract}

%

\section{Introduction}

Restless Multi-Armed Bandits (RMABs) are a powerful
tool for modeling sequential decision-making problems
under resource constraints.
RMABs model a setting where an agent must choose some number of actions at each time step.
Each action incurs a cost and yields a stochastic reward depending on the state of the environment.
Traditionally, the agent's goal is to maximize the reward subject to a budget constraint.

Since exactly solving RMABs is computationally hard \cite{papadimitriou1994complexity}, a common heuristic is to assign a value---called the Whittle index---for each possible action.
Then the heuristic greedily selects actions with the largest Whittle indices.
This Whittle index heuristic has been successfully applied in a variety of domains including healthcare engagement \cite{mate2022field,biswas2021learn,killian2021beyond}, anti-poaching \cite{qian2016restless}, and sustainable energy \cite{iannello2012optimality}.

A limitation of the maximization formulation of RMABs is that it may not be appropriate for settings with a variable budget where the primary goal is to achieve a certain amount of reward.
We give several examples of such settings:

\paragraph{Wildlife Conservation}
A nonprofit organization seeks to rescue an endangered species by re-introducing captive-bred individuals into wild areas.
Each area has a different set of environmental conditions such as food availability, predator abundance, human proximity, and habitat quality.
There is a cost associated with re-introducing the species into each area and a stochastic reward associated with the number of subsequent offspring born from the re-introduced individuals.
The goal is to re-introduce the species into several areas at minimum cost so that a certain number of individuals are born in the wild.
A solution to the maximization problem may require abandoning the project once the budget is exceeded even if the species is still critically endangered.

\paragraph{Energy Management}
A company seeks to reduce their carbon footprint by cutting their energy use.
There are several energy-saving measures they can take such as installing solar panels, upgrading to smart appliances, using motion-sensing lights, and moderating temperature settings.
Each measure has a different cost and a stochastic reward depending on the reliability of the measure and user behavior.
The goal is to reduce energy consumption by a certain amount at minimum cost.
A solution to the maximization problem may require abandoning energy-saving measures once the budget is exceeded even if energy use remains high.

\paragraph{Healthcare}
A medical team seeks to provide care to a sick patient.
There are different treatments they can provide such as surgery, medication, physical therapy, and counseling.
Each treatment has a different cost and a stochastic reward depending on the patient's condition and the treatment's effectiveness and reliability.
The goal is to provide care at minimum cost so that the patient stabilizes to a healthy state.
A solution to the maximization problem may require abandoning treatment once the budget is exceeded even if the patient is still sick.

In \cite{diaz2023flexible}, the authors consider a more flexible version of the maximization problem where the budget is aggregated over multiple time steps.
While their work is a step in the right direction, it still necessitates a hard budget constraint.
In fact, if their aggregated budget is exceeded early, the results may be even worse since they are unable to take actions for the rest of the aggregated period.

The solution we propose is the minimization problem:
The goal of the minimization problem is to achieve a certain amount of reward while minimizing total cost.
In the wildlife conservation example, the minimization problem requires re-introducing individuals at minimum cost until a certain number of offspring are born in the wild.
In the energy management example, the minimization problem requires becoming more energy efficient at minimum cost until energy usage is reduced by a certain amount.
In the healthcare example, the minimization problem requires providing care at minimum cost until the patient recovers to their previous baseline health.
While not applicable to all settings, the minimization problem is more appropriate than the maximization problem when the primary goal is to achieve a certain amount of reward.

\subsection{Our Contributions}

Our first contribution is the formulation of the
minimization problem for RMABs.
Since solving the minimization problem exactly is computationally intractable, we introduce a bi-criteria approximation problem.
We then show that even finding a bi-criteria approximation within any approximation factor is PSPACE-hard.
As a result, if PSPACE $\neq$ P, there are no polynomial-time algorithms which can provably solve the bi-criteria approximation problem within any approximation factor.

Given the computational hardness result, our second contribution is the decoupling of the minimization problem.
Analogous to the maximization problem, we introduce a decoupled problem, a notion of indexability, and a Whittle index for the minimization problem.

Our third contribution is a comparison between the minimization and maximization problems.
We show that the indexability of the maximization problem implies the indexability of the minimization problem, and vice versa.
Further, we show a simple relationship between the Whittle index of the maximization problem and the Whittle index of the minimization problem. 
It then follows that existing results on the Whittle index for the maximization problem give the Whittle index for the minimization problem.
While the minimization and maximization problems are similar in many ways, algorithms designed for the maximization problem do not necessarily perform well on the minimization problem:
We present an RMAB instance where the standard heuristic for the maximization problem gives the optimal strategy but the analogous heuristic for the minimization problem gives the worst possible strategy.

Inspired by the need for new heuristics, our fourth contribution is the development of two heuristics for the minimization problem, inspired by prior work: an increasing budget heuristic and a truncated reward heuristic.
We compare the heuristics on anonymized patient data from the National Inpatient Sample and synthetic data generated from a well-studied hidden two-state RMAB for which the Whittle index is known exactly in closed form \cite{le2008multi,liu2010indexability}.

\subsection{Related Work}

\paragraph{Restless Multi-Armed Bandits}

A line of recent work has studied RMABs
for the problem of promoting patient engagement.
In \cite{mate2022field}, the authors present the results of
a field study on the impact of an RMAB solution to
maternal and child health.
Several papers have also studied the problem of approximating
the Whittle index for RMABs when the transition function is
not known in advance \cite{biswas2021learn,wang2023optimistic}.
In \cite{killian2021beyond}, they extend the RMAB instance
to the case where there are more than two possible actions.
In \cite{wan2015weighted}, they study the maximization problem within the more general setting of non-negative costs.
We similarly consider the setting with non-negative costs but in the context of the minimization problem we introduce.
RMABs have also been studied in the context of 
anti-poaching \cite{qian2016restless}
and sustainable energy \cite{iannello2012optimality}.

\paragraph{Restless Bandits and Exact Whittle Index}

There is a large body of work on the Whittle index for
various RMAB instances.
For (special cases of) the following RMAB instances,
the Whittle index is known in closed form:
a hidden two-state Markov chain where the state
is only learned if the chain is activated
\cite{liu2010indexability,le2008multi},
a two-state Markov chain where the state is always unknown
\cite{meshram2018whittle},
several variants of the age of information problems
\cite{hsu2018age,tripathi2019whittle},
collapsing bandits \cite{mate2020collapsing},
and crawling websites for ephemeral content
\cite{avrachenkov2016whittle}.

\paragraph{Q-Learning Whittle Index}

Often, the transition function is not known in advance
or is too complex to compute exactly.
In this case, 
establishing indexability and deriving closed form equations
for the Whittle index can be challenging.
A line of recent work has 
used Q-learning to approximate Whittle indices
\cite{fu2019towards,nakhleh2021neurwin,killian2021q,avrachenkov2022whittle,robledo2022qwi}.
Q-learning is a model-free reinforcement learning algorithm that learns the value of taking an action in a particular state.
We use Q-learning in our real data experiments to approximate the Whittle index.

\paragraph{General Guarantees of the Whittle Index}

While there are no general guarantees for the optimality
of the Whittle index strategy,
there is a line of work showing that the strategy is optimal
in a limited asymptotic sense \cite{weber1990index,gast2023exponential}
and in special cases \cite{mate2021risk}.
In \cite{akbarzadeh2022conditions}, the authors study general conditions
under which a problem is indexable.

\section{RMAB Definition and Notation}

The Restless Multi-Armed Bandit (RMAB) problem
is built on top of $n$ independent
Markov Decision Processes (MDPs).
Consider a particular MDP $i \in [n]$ with states
$\mathcal{S}_i$ and actions $\mathcal{A}_i$.
Let $\tau_i$ be the transition function that stochastically
maps each state-action pair to a state.
Let $r_i$  be a reward function that stochastically
maps each state-action pair to a real-valued reward.
Finally, let $c_i$ be a deterministic cost function that maps actions to non-negative real costs.
At time step $t$, the agent observes the state $s_i^{(t)} \in \mathcal{S}_i$, selects an action $a_i^{(t)} \in \mathcal{A}_i$, incurs a cost $c_i(a_i^{(t)})$, and receives a reward $r_i(s_i^{(t)}, a_i^{(t)})$.
The MDP then transitions to a new state $s_i^{(t+1)}$ according to the transition function $\tau_i(s_i^{(t)}, a_i^{(t)})$.
Let $\mathcal{S}^n = \mathcal{S}_1 \times \ldots \times \mathcal{S}_n$ be the combined state space and $\mathcal{A}^n = \mathcal{A}_1 \times \ldots \times \mathcal{A}_n$be the combined action space.
We use $\tau$, $r$, and $c$ to denote the transition, reward, and cost functions on the combined state and action spaces.

The problem can be formulated for a general set
of actions but, in order to define the Whittle index, we assume each action space $\mathcal{A}_i$ is binary.
In this case, action $1$ corresponds to choosing MDP $i$ and action $0$ (also called the {\em passive} action) corresponds to not choosing MDP $i$.
The cost function is generally restricted so that $c(1) = 1$ and $c(0)=0$.
However, we consider a more general setting where $c(1)$ could be any non-negative number.
There are several ways to define the objective function of the RMAB problem.
In this work, we consider the discounted expected reward with discount factor $\beta \in (0,1)$, over an infinite horizon.

The RMAB maximization problem, generalized to non-negative costs, is as follows:

\begin{problem}[Exact Maximization]
Consider a budget $C$.
The optimal solution to the maximization problem is a policy
$\pi^+ \colon \mathcal{S}^n \rightarrow \mathcal{A}^n$
with maximum expected discounted reward
\begin{align}\label{eq:max_objective}
    \E_{\{\mathbf{a}^{(t)}, \mathbf{s}^{(t)}\}_{t=1}^\infty \sim \pi, \tau}
    \left[ \sum_{t=1}^\infty \beta^{t-1} \sum_{i=1}^n
    r_i(s_i^{(t)}, a_i^{(t)}) \right]
\end{align}
subject to the constraint that
    $\sum_{i=1}^n c_i(a_i^{(t)}) \leq C$
for all $t$.
\end{problem}

Since evaluating the optimal strategy for the maximization problem
is PSPACE-hard \cite{papadimitriou1994complexity},
the classical approach is to make a series of relaxations
to get a decoupled problem for each MDP.
When we consider a single MDP $i$, we typically drop the subscript for notational brevity.
The decoupled problem considers each MDP in isolation and assigns a value to each action.
Let
\begin{align*}
    V_{\max}(s,a,\lambda) &= \E
    \left[ r(s, a) - \lambda c(a) \right] \\
    &+ \E_{s' \sim \tau(s,a)}
    \left[\beta \max_{a'} V_{\max}(s', a', \lambda) \right].
\end{align*}

The Whittle index gives a way to compare the value of taking the active action and the passive action for each MDP in each state.
(To define the Whittle index, we need a technical condition called indexability, the details of which appear in the technical appendix due to space constraints.)

\begin{definition}[Maximization Whittle Index \cite{whittle1988restless}]
The Maximization Whittle Index for an MDP in state $s$
is the smallest value for which it is optimal to take
the passive action.
Formally, the Whittle Index is given by
\begin{align}
    \lambda^+ = \inf \{\lambda: V_{\max}(s,0,\lambda) > V_{\max}(s,1,\lambda)\}.
\end{align}
\end{definition}

Because a larger value indicates the active action is more valuable than the passive action, the Whittle index suggests a natural measure that we can use to compare the value of taking an action in different MDPs.

\begin{algorithm}[tb]
    \caption{Greedy Maximization Heuristic (\textsc{GreedyMax})}
    \label{alg:max_heuristic}
    \textbf{Input}: A state $\mathbf{s}$, a reward function $r$, a cost function $c$, budget $C$ \\
    \textbf{Parameter}: Discount factor $\beta$ \\
    \textbf{Output}: Selected actions $\mathbf{a}$
    \begin{algorithmic}[1] 
        \STATE $\mathbf{a} \gets 0$ \COMMENT{Initialize selected actions}
        \WHILE{an unselected MDP fits budget}
            \STATE $i^* \gets \arg \max_{i \in [n]: a_i = 0} \{\lambda_i^+ : c_i(1) + c(\mathbf{a}) \leq C \}$
            \STATE $a_{i^*} \gets 1$ \COMMENT{Select MDP $i^*$}
        \ENDWHILE
        \STATE \textbf{return}  $\mathbf{a}$
    \end{algorithmic}
\end{algorithm}

In the classical RMAB maximization problem, the constraint is to choose a fixed number $m$ of actions in every time step to be active.  
A standard heuristic is to choose the $m$ bandits with the highest Whittle indices.
A simple generalization of this heuristic to arbitrary costs is given in Algorithm~\ref{alg:max_heuristic}.  
In each step, it chooses the MDP with highest Whittle index, from all MDPs that fit the remaining budget. 

\section{The Minimization Problem}

We introduce the RMAB minimization problem for settings where a fixed amount of reward must be met.

\begin{problem}[Exact Minimization]
Consider a reward threshold $R$.
The solution to the minimization problem is a policy
$\pi^- \colon \mathcal{S}^n \rightarrow \mathcal{A}^n$
with minimum expected discounted cost
\begin{align}\label{eq:rmab_min_objective}
    \E_{\{\mathbf{a}^{(t)}, \mathbf{s}^{(t)}\}_{t=1}^\infty \sim \pi, \tau} \left[
    \sum_{t=1}^\infty \beta^{t-1} \sum_{i=1}^n
    c_i(a_i^{(t)})
    \right]
\end{align}
subject to the constraint that
    $\sum_{i=1}^n r_i(s_i^{(t)}, a_i^{(t)}) \geq R $
for all $t$.
\end{problem}

It is easy to satisfy the budget constraint of the maximization problem by simply limiting the number of actions selected.
In contrast, it is difficult to satisfy the reward constraint of the minimization problem because the rewards from actions are stochastic.
As a result, the challenge of the minimization problem stems
from \textit{both} minimizing the objective \textit{and} satisfying the constraint.
We therefore consider a bi-criteria approximation problem.

\begin{problem}[Approximate Minimization]
Consider an approximation factor $\alpha \geq 1$ 
and a success probability $\rho \in (0,1]$.
A policy $\pi$ is an $(\alpha, \rho)$-approximation
to the minimization problem if its expected
discounted cost (Equation \ref{eq:rmab_min_objective})
is within a factor of $\alpha$ of the optimal policy $\pi^-$
and
\begin{align}
    \Pr \left(\sum_{i=1}^n r_i(s_i, a_i^{(t)}) \geq R \right)
    \geq \rho
\end{align}
for all $t$.
\end{problem}

Note that the approximate minimization problem is allowed
to violate the reward constraint with probability $1-\rho$
but only with respect to the randomness of the reward function.

Our main theoretical result is that even the approximate minimization problem is computationally hard.

\begin{restatable}{theorem}{pspace}
\label{thm:pspace}
Fix $\alpha \geq 1$ and $\rho > 0$.
Finding an $(\alpha,\rho)$-approximate strategy
for the minimization problem is PSPACE-hard even when costs are binary.
\end{restatable}

\begin{proof}[Proof Outline]
Our reduction is from the generic problem of determining whether a polynomial-space Turing machine halts.
The reduction constructs an  RMAB minimization instance which has an MDP  corresponding to each cell of the Turing machine tape.  
A policy for the RMAB instance can either simulate the Turing machine, or not.
If it simulates the Turing machine, it incurs cost at least $2\alpha^2$ if the Turing machine halts and no cost otherwise.
If it does not simulate the Turing machine, it always incurs cost $\alpha$.
Therefore making the optimal choice requires determining whether the Turing machine halts and making the wrong choice gives a worse than $\alpha$ approximation (to satisfy the reward constraint with any non-zero probability since the reward function is deterministic).
\end{proof}

Because of its length, we delay the full proof of the theorem to the technical appendix.
In the proof, we show that an algorithm which (even approximately) solves the minimization problem can be used to determine whether a polynomial-space Turing machine halts.
Because determining whether a polynomial-space Turing machine halts is PSPACE-complete by definition, it immediately follows that the approximate minimization problem is PSPACE-hard.

\section{Decoupling the Minimization Problem}

Given the hardness of the approximate minimization problem, there are no efficient algorithms with provable guarantees unless P $=$ PSPACE.
Instead, we turn to heuristics.
Following the approach of \cite{whittle1988restless}, we decouple the minimization problem to consider each MDP in isolation.
In particular, we relax the constraint in the minimization problem and then apply Lagrange multipliers.
The resulting decoupled problem for a particular MDP in the minimization problem appears below.
For notational brevity, we drop the subscript $i$.

\begin{problem}[Decoupled Minimization]
Consider a particular MDP in the RMAB instance.
Fix $\lambda \geq 0$.
The decoupled minimization problem for that MDP is to find the policy
\begin{align}
    \arg \max_{\pi \colon \mathcal{S} \rightarrow \mathcal{A}}
    \E \left[ \sum_{t=1}^\infty \beta^{t-1}
    (\lambda \cdot r(s^{(t)}, a^{(t)}) - c(a^{(t)})) \right].
\end{align}
\end{problem}

\begin{proof}[Proof of Decoupling]
We now show how to convert the exact minimization problem
into the decoupled minimization problem.
The process is analogous to the maximization case in the literature \cite{whittle1988restless,wang2023optimistic}.
The idea is to turn the constrained optimization problem
into an unconstrained optimization problem.
We will accomplish this by applying Lagrange multipliers.
However, the first step is to relax the constraint so that
the objective and constraint are similar.
In particular, we will relax the constraint 
in the exact minimization problem to hold on average:
\begin{align}
    (1-\beta) \E \left[
    \sum_{t=1}^\infty \beta^{t-1} \sum_{i=1}^n
    r_i(s_i^{(t)}, a_i^{(t)}) \right]  \geq R.
\end{align}
The multiplicative normalization factor $1-\beta$
is chosen so that if the strict constraint is satisfied 
for all $t$, 
then the relaxed constraint is also satisfied.


We apply Lagrange multipliers to the constrained problem under the relaxed constraint and reach the Lagrangian function given by:
\begin{align}\label{eq:after_lagrange}
    \E \left[
    \sum_{t=1}^\infty \beta^{t-1} \sum_{i=1}^n
    c_i(a_i^{(t)}) - \lambda \cdot r_i(s_i^{(t)}, a_i^{(t)}) \right]
    + \frac{\lambda R}{1-\beta}
\end{align}
where $\lambda \geq 0$ is a Lagrange multiplier.


The next step is to decouple the unconstrained problem.
We can interchange the summations since $n$ is finite.
Then we consider the problem for a fixed $\lambda \geq 0$.
The resulting problem is
\begin{align}
    \min_{\pi} \E \left[ 
    \sum_{i=1}^n \sum_{t=1}^\infty \beta^{t-1}
    (c_i(a_i^{(t)}) - \lambda \cdot r_i(s_i^{(t)}, a_i^{(t)}))
    \right]
    + \frac{\lambda R}{1-\beta}.
\end{align}
The problem is now decoupled since the policy for each
MDP is optimized in isolation.
With the observations that the final term is constant
for fixed $\lambda$ and that minimization is equivalent
to maximization after a sign flip, 
the decoupled minimization problem follows.
\end{proof}

The difficulty of the RMAB problem lies in the complicated interactions between MDPs.
By considering each MDP separately, the decoupled problem lets us characterize the `value' of selecting a particular MDP.
Then the RMAB problem can be solved by a heuristic for choosing MDPs that relies on their value.

Analogous to the maximization problem, we introduce the Whittle index for the minimization problem to characterize the value of choosing each MDP.
We will first define the value function for the minimization problem.
\begin{align*}
    V_{\min}(s,a,\lambda) &= \E
    \left[ \lambda \cdot r(s, a) - c(a) \right] \\
    &+ \E_{s' \sim \tau(s,a)}
    \left[\beta \max_{a'} V_{\min}(s', a', \lambda) \right].
\end{align*}

For the Whittle index to be defined, the MDP needs to meet a technical condition known as indexability.

\begin{definition}[Indexability for Minimization]
An MDP is indexable if for all states $s \in \mathcal{S}$
and real numbers $\lambda' \leq \lambda$,
\begin{align}
    V_{\min}(s,0,\lambda) &> V_{\min}(s,1,\lambda) \nonumber \\
    \implies
    V_{\min}(s,0,\lambda') &> V_{\min}(s,1,\lambda').
\end{align}
\end{definition}

In words, indexibility is the following property: If it is optimal to take the passive action under
a subsidy $\lambda$ in state $s$, 
then it must also be optimal to take the passive 
action under a smaller subsidy $\lambda'$ in state $s$.

Then the Whittle index for the minimization problem is naturally the following:

\begin{definition}[Whittle Index for Minimization]
The minimization Whittle Index for an MDP in state $s$
is the largest value for which it is optimal to take
the passive action.
Formally, the index is given by
\begin{align}
    \lambda^- = \sup \{\lambda: V_{\min}(s,0,\lambda) > V_{\min}(s,1,\lambda)\}.
\end{align}
\end{definition}

\begin{algorithm}[tb]
    \caption{{Greedy Minimization Heuristic}}
    \label{alg:min_heuristic}
    \textbf{Input}: A state $\mathbf{s}$, a reward function $r$, and a cost function $c$ \\
    \textbf{Parameters}: Discount factor $\beta$, reward threshold $R$, success probability $\rho$ \\
    \textbf{Output}: Selected actions $\mathbf{a}$
    \begin{algorithmic}[1]
    \STATE $\mathbf{a} \gets 0$ \COMMENT{Initialize selected actions}
    \WHILE{$\Pr\left(\sum_{i:a_i=1} r_i(s_i, 1) \geq R\right) < \rho \land \| \mathbf{a} \|_1 < n$}
        \STATE $i^* \gets \arg \min_{i: a_i = 0} \lambda_i^- $
        \STATE $a_{i^*} \gets 1$ \COMMENT{Select MDP $i^*$}
    \ENDWHILE
    \end{algorithmic}
\end{algorithm}
With the Whittle index in hand, we can develop a heuristic for the minimization problem.
The heuristic is analogous to the greedy maximization heuristic except that the stopping condition is different.
Instead of stopping when the budget is exceeded, the minimization heuristic stops when the reward constraint is probabilistically satisfied.
Algorithm~\ref{alg:min_heuristic} presents this strategy in pseudocode.

\paragraph{Remark:} It is not obvious how to determine the probability of satisfying the constraint.
If the reward function is known in closed form,
then the probability can be computed exactly.
However, in many cases, the reward function is not known in closed form or computing the exact probability is computationally intensive (i.e., because there are many possibilities on the combined outcome space of the actions).
Another option is to use a concentration inequality
specialized to sums of random variables such as
Bernstein's or Hoeffding's inequalities \cite{bernstein1924modification,hoeffding1963probability}.
However, the concentration inequality may be quite loose depending on the RMAB reward function and so the heuristic could be overly conservative.
The option we recommend is to simulate a small number of realizations of the reward function.
This approach is computationally efficient and can be made arbitrarily accurate by increasing the number of simulations.
However, simulations may not be possible in all settings of the problem.

%
%

\section{Maximization vs Minimization Problems}

In this section, we show the close connection between the Whittle indices for the maximization and minimization problems.
However, despite this connection, we also show that heuristics analogous to those used for the maximization problem can perform arbitrarily poorly for the minimization problem.

In general, indexability does not hold for all
RMAB instances \cite{weber1990index}.
So the first step of using any Whittle index-based heuristic is establishing that indexability holds.
Unfortunately, it can be quite difficult for a particular RMAB instance.
Fortunately, we show that if indexability
holds for the maximization problem
then it also holds for the minimization problem.
Due to space constraints, the proofs of Corollaries \ref{coro:indexability} and \ref{coro:index} appear in the technical appendix.

\begin{restatable}{corollary}{indexability}
\label{coro:indexability}
The decoupled maximization problem
is indexable if and only if the 
decoupled minimization problem is indexable.
\end{restatable}

A similar result holds for the Whittle index.
If the Whittle index is known for either the
maximization or minimization problem then it
is also known for the other problem.

\begin{restatable}{corollary}{index}
\label{coro:index}
Suppose the decoupled maximization
and minimization problems are indexable.
Let $\lambda^+$ be the Whittle index for
the decoupled maximization problem
and $\lambda^-$ be the index for the decoupled
minimization problem.
If $\lambda^+, \lambda^- > 0$, then $\lambda^+ = 1/\lambda^-.$
\end{restatable}

Instead of deriving the Whittle index for the minimization problem \textit{and} the maximization problem, Corollary \ref{coro:index} tells us how to find the Whittle index for the minimization problem if the Whittle index for the maximization problem is already known.
The following are a selection of RMABS where the Whittle index is known in closed form for the maximization problem: age of information* \cite{tripathi2019whittle}, partially hidden two-state \cite{liu2010indexability}, completely hidden two-state \cite{meshram2018whittle}, collapsing bandits \cite{mate2020collapsing}, crawling content* \cite{avrachenkov2016whittle}, and controlled resets \cite{akbarzadeh2019restless}.
An asterisk indicates the problem is formulated to maximize the expected (rather than discounted) reward.
By Corollaries \ref{coro:indexability} and \ref{coro:index}, the minimization problem is indexable for each of these instances and the Whittle index for the minimization problem can be easily computed.


So far, it seems that the maximization
and minimization problems are morally the same.
However, we will show that algorithms adapted from the maximization problem can perform arbitrarily poorly for the minimization problem.

\begin{claim}\label{claim:min_vs_max}
Let $n$ be the number of MDPs and
$\rho > 0$ be a success probability.
There is a simple RMAB instance with unit costs where
Algorithm~\ref{alg:max_heuristic} is optimal but Algorithm~\ref{alg:min_heuristic} has expected cost $n$ times the optimal strategy in order to satisfy the constraint with probability $\rho$.
\end{claim}

Notice that the performance of  Algorithm~\ref{alg:min_heuristic} is the worst possible for the unit cost case:
Always choosing \textit{every} MDP in each time step 
will trivially give an $n$-approximation.

\begin{proof}[Proof of Claim \ref{claim:min_vs_max}]
Consider the following RMAB instance.
For every MDP $i \in [n]$, there is a single state $s$.
If the active action is selected, then
with probability $0 < p_i < 1$, reward $r_i > 0$ is received. 
The cost of selecting the active action from $s$ is $1$.
If the passive action is selected,
then no reward is received and no cost is incurred.
Observe that for MDP $i$, we have
\begin{align*}
    V_{\max}(s,1,\lambda) &= p_i r_i - \lambda
    + \max_{a} V_{\max}(s,a,\lambda) \\
    V_{\max}(s,0,\lambda) &= 0
     + \max_{a} V_{\max}(s,a,\lambda).
\end{align*}
The decoupled maximization problem is clearly indexable
and the Whittle index is $\lambda_i^+ = p_i r_i$.
Algorithm~\ref{alg:max_heuristic} selects MDPs
with the largest Whittle indices first.
For the maximization problem, this strategy is optimal
(a simple interchange argument shows why).

By Corollary \ref{coro:indexability} 
and Corollary \ref{coro:index}, 
the decoupled minimization problem is also indexable
and the Whittle index for the minimization problem is
$\lambda_i^- = 1 / (p_i r_i)$.
Algorithm~\ref{alg:min_heuristic} selects MDPs
with the smallest minimization Whittle indices 
(i.e., largest maximization Whittle indices) first.

We now exhibit a choice of parameters where the strategy
fails for the minimization problem.
For $i < n$, let $p_i = \frac{\log_{1/(2e)}(1-\rho)}{n-1}$
and $r_i=10 R /p_i$ 
where $R$ is the reward threshold.
Then the Whittle index for the minimization problem is $\lambda_i^- = 1/(10R)$.
Let $p_{n} = 1$ and $r_{n} = R$.
Then the Whittle index for the minimization problem is $\lambda_{n}^- = 1/R$.
Algorithm~\ref{alg:min_heuristic} selects
the MDPs with $i<n$ first.
Even if the algorithm selects all MDPs with $i<n$,
the probability of \textit{not} satisfying the constraint is
\begin{align}
    \left(
    1 - \frac{\log_{1/(2e)}(1-\rho)}{n-1}
    \right)^{n-1}
    &> \left(\frac{1}{2e}\right)^{\log_{1/(2e)}(1-\rho)}
    = 1-\rho
\end{align}
where the inequality holds for sufficiently large $n$.
Therefore, the probability of satisfying the constraint
is strictly less than $\rho$.
In contrast, choosing the $n$th MDP deterministically
achieves reward $R$ with cost $1$.
Therefore, Algorithm~\ref{alg:min_heuristic}
has expected cost $n$ 
times the optimal strategy in order to satisfy the constraint
with probability $\rho$.
\end{proof}

\section{Heuristics for the Minimization Problem}

The example in Claim \ref{claim:min_vs_max} shows that the greedy minimization heuristic described in Algorithm~\ref{alg:min_heuristic} can fail badly.
As a result, we need alternative algorithms for the minimization problem.
Because the approximate minimization problem is PSPACE-hard, there is no polynomial-time algorithm that can provably approximate the minimization problem unless P $=$ PSPACE.
Instead, we can at most hope for \textit{heuristic} algorithms without provable guarantees that perform well in practice.
In this section, we present and discuss two such algorithms we generalize from prior work: a standard increasing budget heuristic described in Algorithm~\ref{alg:budget_min_heuristic} and a more specialized truncated reward heuristic described in Algorithm~\ref{alg:truncated_min_heuristic}.
A slightly simpler version of Algorithm~\ref{alg:truncated_min_heuristic} has been theoretically analyzed in special cases of our problem in prior work \cite{ene2018approximation,jiang2020algorithms,ghuge2022non}.

The first algorithm we consider is an increasing budget heuristic.
At each phase, the heuristic greedily selects MDPs until it exhausts the current budget.
The budget grows exponentially with a multiplicative factor $m$.
In this way, the heuristic takes low cost actions first.
If the low cost actions satisfy the reward constraint, then we've satisfied the constraint at minimum cost.
If the low cost actions do not satisfy the constraint and we need to keep going, then we haven't paid too much more than the optimal strategy because the actions are low cost and the budget grows exponentially.

\begin{algorithm}[h!]
    \caption{Increasing Budget Minimization Heuristic}
    \label{alg:budget_min_heuristic}
    \textbf{Input}: A state $\mathbf{s}$, a reward function $r$, and a cost function $c$ \\
    \textbf{Parameters}: Discount factor $\beta$, reward threshold $R$, success probability $\rho$, budget multiplier $m$ \\
    \textbf{Output}: Selected actions $\mathbf{a}$
    \begin{algorithmic}[1]
    \STATE $\mathbf{a} \gets \mathbf{0}$ \COMMENT{Initialize selected actions}
    \STATE $b \gets \min_{i \in [n]} c_i(1)$ \COMMENT{Budget}
    \WHILE{$\Pr\left(\sum_{i:a_i=1} r_i(s_i, 1) \geq R\right) < \rho \land \| \mathbf{a} \|_1 < n $}
    \STATE $\triangleright$ Greedy {selection\footnotemark} within budget $b$
    \STATE $\mathbf{a'} \gets $ \textsc{GreedyMax}($\mathbf{s}, r, c, b)$
    \STATE $\mathbf{a} \gets \mathbf{a} + \mathbf{a'}$
    \STATE $b \gets mb $ \COMMENT{Update budget}
    \ENDWHILE 
    \end{algorithmic}
\end{algorithm}
\footnotetext{By Corollary \ref{coro:index}, the largest Whittle indices for the maximization problem are the smallest Whittle indices for the minimization problem. We use this observation to simplify the pseudocode by calling \textsc{GreedyMax}.}

\begin{algorithm}[h!]
    \caption{Truncated Reward Minimization Heuristic}
    \label{alg:truncated_min_heuristic}
    \begin{flushleft}
    \textbf{Input}: A state $\mathbf{s}$, a reward function $r$, a cost function $c$ \\
    \textbf{Parameters}: Discount factor $\beta$, reward threshold $R$, success probability $\rho$, budget multiplier $m$ \\
    \textbf{Output}: Selected actions $\mathbf{a}$
    \end{flushleft}
    \begin{algorithmic}[1]
    \STATE $R_{\max} \gets \max_{i} r_i(s_i, 1)$ \COMMENT{Maximum possible reward}   
    \STATE $\mathbf{a} \gets \mathbf{0}$
    \STATE $b \gets \min_{i \in [n]} c_i(1)$
    \WHILE{$\Pr\left(\sum_{i:a_i=1} r_i(s_i, 1) \geq R\right) < \rho \land \| \mathbf{a} \|_1 < n $}
        \FOR{$\tau = 0, \ldots, \lceil \log_2 R_{\max} \rceil$}
        \STATE $\triangleright$ Truncate reward function
        \STATE $r'(\cdot) \gets \min\{r(\cdot)/2^\tau, 1\}$
        \STATE Greedy selection with truncated reward 
        \STATE $\mathbf{a}_\tau \gets $ \textsc{GreedyMax}$(\mathbf{s}, r', c, b)$
        \ENDFOR
    \STATE $\tau^* \gets$ smallest $\tau$ such that $\mathbf{a}_\tau$ is poor\footnotemark
    \STATE $\mathbf{a} \gets \mathbf{a} + \mathbf{a}_{\tau^*}$
    \STATE $b \gets m b $ \COMMENT{Update budget}
    \ENDWHILE
    \end{algorithmic}
\end{algorithm}
\footnotetext{We say a scale is \textit{poor} if all the remaining MDPs that were not selected have a Whittle index of at most $1/b$.}

While it takes low cost actions first, Algorithm~\ref{alg:budget_min_heuristic} can choose actions which are desirable in expectation but only because their reward is large enough to balance out their small probability of having reward.
Notice these actions \textit{are} desirable for the maximization problem because the goal is to maximize \textit{expected} reward.
However, for the minimization problem, these actions are not desirable because they have a low probability of satisfying the reward threshold.

The second algorithm we consider addresses this problem by \textit{truncating} rewards at different levels.
Just as Algorithm~\ref{alg:budget_min_heuristic} initially only considers actions with low cost, Algorithm~\ref{alg:truncated_min_heuristic} initially only considers actions with low reward.
The advantage is that the actions selected have high probability of outputting reward which is helpful for solving the minimization problem.
The algorithm still performs well when high reward actions are better because the truncation factor exponentially increases.
Since we want to keep the increasing budget property, we repeat the truncation factor search for each size of the budget.

\section{Experiments}\label{sec:experiments}

We test our algorithms on real and synthetic data sets, with deterministic costs and stochastic rewards.
On each data set, we compare the discounted cost as the reward threshold, number of MDPS, and success probability vary.
We set the discount factor to $\beta=.9$ and run each simulation for 10 time steps, repeating 10 times.
We report the mean (lines) and standard deviations (shaded regions) in the plots.
The code is available in the supplementary material and will be accessible online after publication.

\begin{figure}[t]
    \centering
    \includegraphics[width=\linewidth]{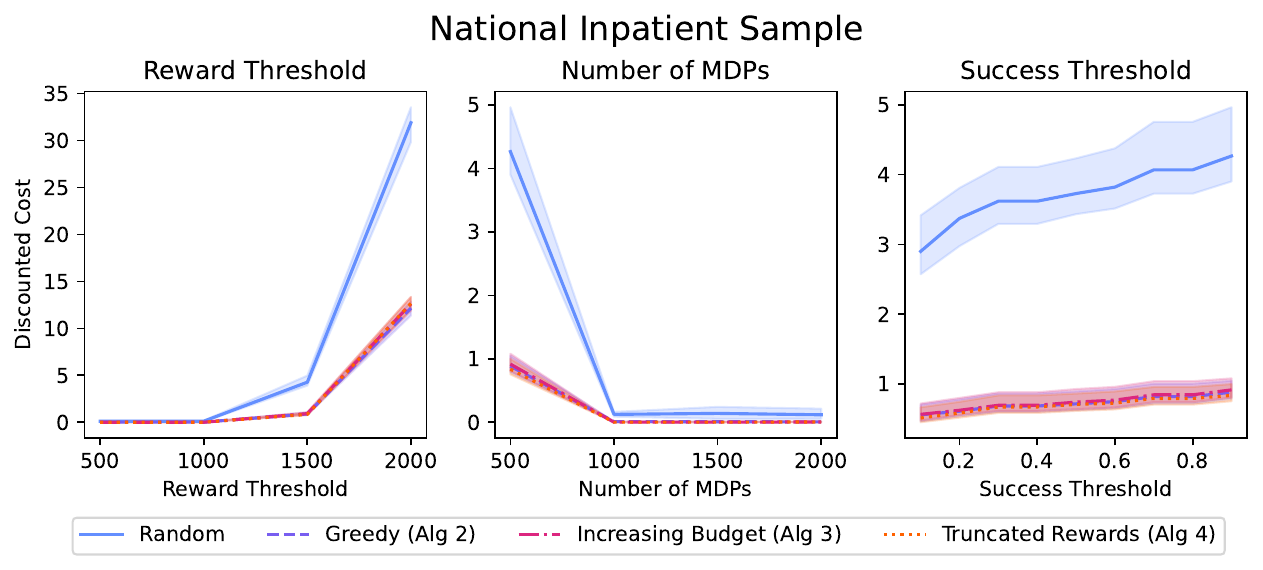}
    \caption{Because the costs and rewards are similar in the dataset, the exponentially increasing budget and the reward truncation have little effect.
    As a result, Algorithms~\ref{alg:min_heuristic}, \ref{alg:budget_min_heuristic}, and \ref{alg:truncated_min_heuristic} all give solutions with similar cost.} 
    \label{fig:real}
    \vspace{-1em}
\end{figure}

\subsection{National Inpatient Sample}

Figure \ref{fig:real} shows the performance of algorithms for selecting patient care in the National Inpatient Sample data set \cite{hcup2018national}.
The cost is the (normalized) dollar cost of treatment and the reward is the improvement in a patient's medical condition as measured by a four-level severity index \cite{averill2003all}.
Because the costs and rewards are well-concentrated in the data, the increasing budget and truncating reward techniques have little effect on the discounted cost.
The random baseline of uniformly selecting new actions is slow because it needs to select many more actions before (probabilistically) satisfying the reward constraint.
Since Whittle indices cannot be computed for the data set in closed form, we use a Q-learning approach to approximate the Whittle indices \cite{robledo2022qwi}.
Additional details are available in the supplementary material.



\subsection{Partially Observable MDPs}

We also test the algorithms on synthetic data sets where the Whittle indices can be computed in closed form.
Each MDP is in either a reward-producing state or a non-reward-producing state.
The MDP transitions between states at each time step and the current state can only be observed if the MDP is selected.
The goal is to select MDPs that are likely to give large reward.

Figure \ref{fig:synthetic_adversarial} shows the performance of algorithms on adversarial instances where half the MDPs reliably give a small reward and half the MDPs unreliably give a large reward.
Algorithms \ref{alg:min_heuristic} and \ref{alg:budget_min_heuristic} (the algorithms are actually the same because the costs are all 1) select the MDPs with unreliable reward first and perform poorly.
In contrast, the truncated reward heuristic quickly gives solutions with lower cost because it truncates the large rewards of the second group and selects the MDPs with reliable reward instead.

Figure~\ref{fig:synthetic_uniform} shows the performance of algorithms on uniform instances where the probabilities are selected randomly while the rewards are chosen so that all MDPs have roughly equal expected reward.
Since the probabilities and therefore rewards are similar, Algorithm~\ref{alg:min_heuristic} gives the best performance.

\begin{figure}[h!]
    \centering
    \includegraphics[width=\linewidth]{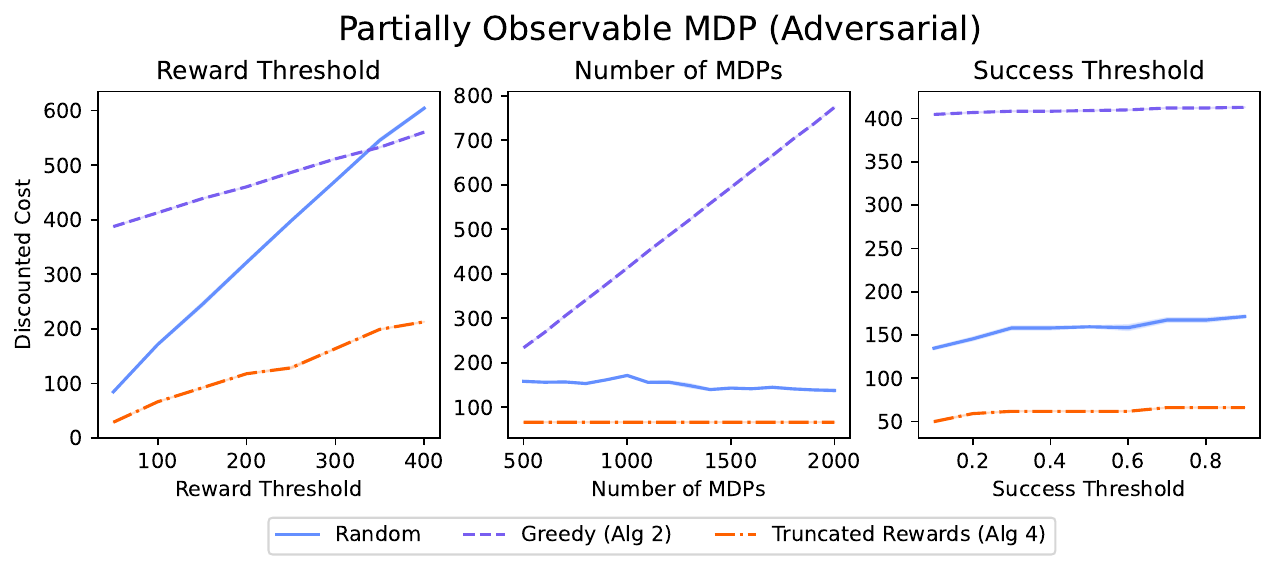}
    \caption{Because the \textit{expected} reward of unreliable MDPS is higher, Algorithms~\ref{alg:min_heuristic} and \ref{alg:budget_min_heuristic} select them first. However, Algorithm \ref{alg:truncated_min_heuristic} quickly gives solutions with lower cost because it truncates the large rewards of the second group and selects the MDPs with reliable reward instead.}
    \label{fig:synthetic_adversarial}
    \vspace{-1em}
\end{figure}

\begin{figure}[h!]
    \centering
    \includegraphics[width=\linewidth]{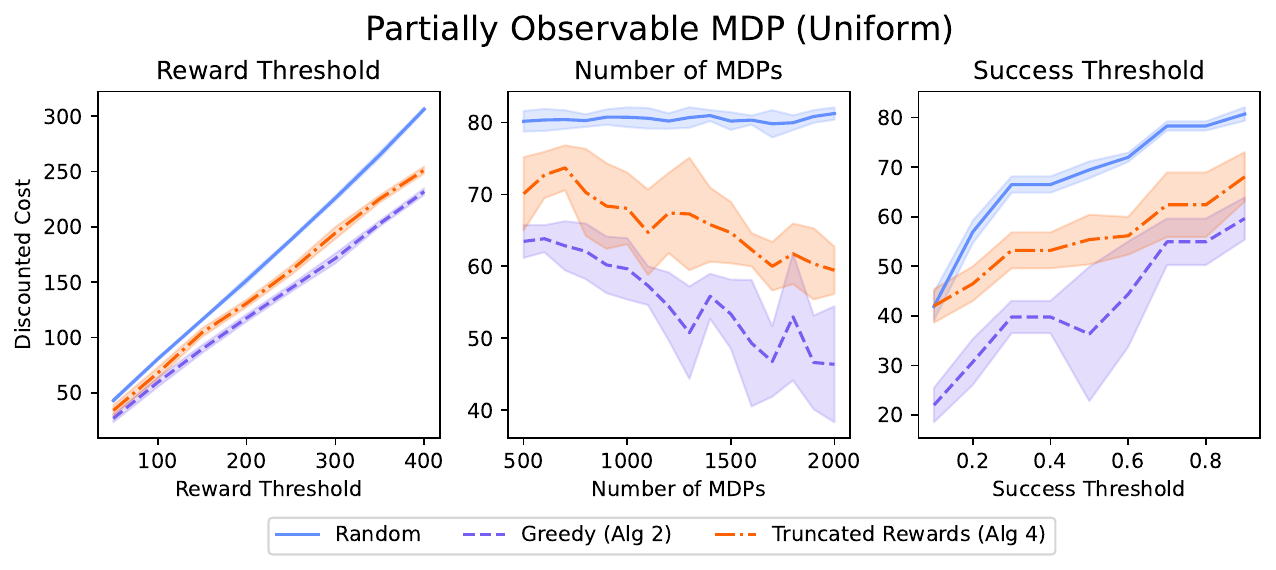}
    \caption{Since the probabilities and rewards are similar in the uniform parameter setting, expected reward is a good indication of quality and Algorithm~\ref{alg:min_heuristic} gives the best performance.}
    \label{fig:synthetic_uniform}
    \vspace{-1em}
\end{figure}

\section{Conclusion}

We introduce the minimization problem for RMABs, designed specifically for applications where a certain amount of reward must be achieved.
We show that even approximating the minimization problem is PSPACE-hard.
Without provably accurate algorithms, we turn to heuristics for solving the minimization problem, defining minimization Whittle indices and presenting two heuristics. 
While algorithms adapted from the maximization problem perform well when rewards are well-concentrated, we find that specialized algorithms are needed for more complex problems.
We believe our work suggests the importance of continued research into the minimization problem for RMABs. 

\section*{Acknowledgements}
R. Teal Witter was supported by the NSF Graduate Research Fellowship under Grant No. DGE-2234660.
Lisa Hellerstein was supported in part by NSF Award IIS-1909335.

\clearpage
\bibliographystyle{alpha}
\bibliography{references}

\appendix

\clearpage
\section{Proof of PSPACE-hardness}

\pspace*

\begin{proof}[Proof of Theorem \ref{thm:pspace}]
Let $M$ be a 1-tape Turing Machine that on any input $x$, runs in space $p(|x|)$, for some polynomial $p$ (where $|x|$ is the length of $x$).  We reduce the problem of determining whether $M$ halts on a given input $x$ to the problem of determining an approximately optimal policy for a given RMAB minimization instance.

The reduction is as follows.
Given $x$, let $n=p(|x|)$,
the maximum number of tape cells that will be used by $M$ when run on input $x$.  We assume without loss of generality that the tape of $M$ consists of $n$ cells, numbered 1 through $n$, that $x$ appears at the start off the tape, and that the head of $M$ only accesses cells 1 through $n$ when run on input $x$.

Let 
$\Gamma$ be the tape alphabet of $M$,
$\Sigma \subseteq \Gamma$ be the tape alphabet,
$Q$ be the set of states of $M$, 
$q_0 \in Q$ be the initial state of $M$,
$q_{accept}, q_{reject} \in Q$ be the halting states,
and $\delta \colon Q \times \Gamma
\rightarrow Q \times \Gamma \times \{R,L\}$.

Let $T = |Q| n^{|\Gamma|+2}$, which is an upper bound on the number of steps that $M$ could perform on input $x$ if $M$ halts on $x$. 

We now describe the RMAB minimization instance.
We choose $\beta$ so that the discounted costs will be between 1/2 and 1 for each step of the RMAB.
In particular, we set $\beta = \exp(\log(.5)/(T+5\alpha^2))$.
The  problem is to find a strategy for the RMAB that gives an $(\alpha, \rho)$-approximation to the minimization problem with reward threshold $R$.

We will describe the instance of the RMAB minimization problem 
for which solving the approximate minimization problem
requires determining whether Turing Machine $M$ halts in $T$ steps.
Interestingly, the MDPs in the RMAB are deterministic.

In the RMAB, there are $n+1$ MDPs, a  ``special'' MDP, and $n$ ``cell'' MDPs, one for each cell of the tape of $M$.

The operation of the RMAB can be described has having two phases: first a warm up phase, and then a simulation phase, in which the cell MDPs can simulate the operation of Turing Machine $M$ on input $x$.

To perform the simulation of the Turing Machine, each cell MDP has states that can keep track of 
the current Turing machine state (if the head is at that cell, or about to move to that cell), the current symbol contained in the tape cell,  
and whether or not the Turing Machine's head is currently at that cell.  

During the warm up phase of the RMAB, there are a series of deterministic transitions between the cell MDP states, with zero cost and reward, that end in the proper initialization of the cell MDPs to represent the Turing machine's initial tape contents, head position, and state.

For MDP $i\in [n]$ at time step $t \in [T]$ of the Turing Machine's simulation,
the state of the MDP is represented by a 7-tuple:
\begin{align}\label{eq:reduction_state}
\left(i, \text{TMst}_i^{(t)}, \text{symbol}_i^{(t)}, 
\text{current}_i^{(t)}, \text{next}_i^{(t)}, j^{(t)}, k^{(t)} \right).
\end{align}
When clear from context, we drop the subscript and superscript.
Here $\text{TMst} \in Q \cup \{0\}$
represents the current state of the Turing Machine if the head is currently at the cell (or about to move there), or 0 otherwise.  Additionally,
$\text{symbol} \in \Gamma$ represents the current symbol in the cell,
$\text{current} \in \{ \texttt{True}, \texttt{False} \}$ represents
whether the cell currently has the head, and
$\text{next} \in [n]$ represents the cell that the head will move to next.
(This value is irrelevant when the cell does not contain the head.)
The indices $j \in [n]$ and $k \in [2|Q|]$ are artifacts of the reduction
used for passing the head and its current state between cells.
We can think of these indices as a central `clock' which allows for communication between MDPs \textit{even though} they are independent. This approach was used in proving PSPACE hardness of the RMAB maximization problem \cite{papadimitriou1994complexity}.

After the warm up phase described below, the indices $j$ and $k$ of the cell MDPs are initialized at 0.
At each time step of the RMAB during the simulation phase, 
each cell MDP updates indices $j$ and $k$ as follows: $$j^{(t)} = j^{(t-1)} + 1 \mod n$$ and, if $j^{(t-1)} = 0$,
$$k^{(t)} = k^{(t-1)} + 1 \mod 2 |Q|+1.$$
In this way, the indices $j$ are incremented $n$ times for each time
$k$ is incremented.
The remaining values stay fixed unless otherwise noted.

If the simulation ends (the Turing machine halts), the MDP corresponding to the cell containing the head of the Turing Machine enters a trapping state that incurs cost $1$ and provides reward $R$.

Before describing the simulation phase in more detail, we first describe the warm up phase.  It has $2\alpha$ time steps.  Both the cell MDPs and the special MDP use their states to keep track of the number of time steps spent so far in this phase.  The $n$ cell MDPs incur zero cost and earn zero reward throughout the steps of the warm up phase, no matter whether or not they are played.  

At the start of the warm up phase, the RMAB policy needs to decide whether or not to play the special MDP.  If it does, the special MDP enters a series of states 
where for the next $2\alpha$ steps, it incurs unit cost and reward $R$ if it is played, and zero cost and zero reward if it is not played.  Beyond that, it enters a trapping stage with zero cost and zero reward.

If the special MDP is not played at the first step, it enters a series of $2\alpha$ states where it receives reward $R$ at zero cost whether or not it is played.
After the $2\alpha$ states, the special MDP then transitions to a trapping state that incurs zero cost and receives zero reward.

Since the constraint in the RMAB minimization problem requires at least $R$ reward to be received in every time step,
if the RMAB policy plays the special MDP at the first step of the initialization phase, it must play the
 special MDP throughout the phase, since that is the only way to satisfy the constraint.  
 Hence by the choice of $\beta$, 
 if the special MDP is played in the first step,
 the discounted cost incurred in the initialization phase is between $\alpha$ and $2 \alpha$, and no further cost needs to be incurred.
 

If the special MDP is not played in the first step, then the only way to receive reward following the initialization phase (and to satisfy the reward constraint) is to simulate the Turing Machine by playing the cell MDPs.
We describe below how the simulation works, but at this point what is relevant is simply that during the simulation, no cost is incurred.  If and when the Turing Machine halts, it must do so within $T$ steps, the RMAB simulatiom will end, and the RMAB will incur a (discounted) cost of at least 
$2\alpha^2$.
Therefore the optimal policy is to play the special MDP if and only if the Turing machine halts on input $x$.
Every other strategy is at least a multiplicative factor of $\alpha$ worse than the optimal strategy.
Then we have that any $(\alpha, \rho)$-approximation to the optimal strategy requires determining whether the Turing machine halts in $T$ steps.
Since the RMAB is deterministic, the statement holds even for $\rho$ arbitrarily close to 0.

We now describe the transition, reward, and cost functions for the cell MDPs that ensure that the Turing Machine is simulated properly during the simulation phase.

Each step of the Turing machine is simulated by a full round of updates
to the indices $j$ and $k$ (that is, $n (2 |Q|+1)$ steps of the MDPs).
The difficulty lies in correctly copying the information about the Turing Machine state from the cell MDP for current head position
\textit{only} to the cell MDP for the next head position.
The index $j$ corresponds to the possible current position of the head (and hence which cell MDP should be copied), while the index $k$ corresponds to the possible current state of the Turing Machine.
The simulation of a single step of the Turing Machine corresponds to multiple steps of the RMAB, consisting of a
transition phase, copy phase, and a validation phase.

In the transition phase, the cell currently with the head is wiped and the next cell with the head is initialized.
The initialization includes computing the new state of the Turing Machine, the new symbol, and the next head position.

In what follows, we give pseudocode 
describing the operation of 
cell MDP $i$, corresponding to the $i$th cell of the Turing Machine tape.

Each time step of the RMAB corresponds to a single value of $j$ and $k$.
The cell RMABs all have the same values for $j$ and $k$ at each time step,
so that there are essentially $n$ copies of each phase running in parallel, one for each cell MDP.

For most values of $j$ and $k$, the reward for a cell MDP is the same whether or not the MDP is played.  (The cost is always 0.)  In this case, we simply give the reward value.  If the reward is received only if the MDP is played, we indicate that explicitly
(otherwise, assume there is no reward).
We use $H=\{q_{accept},q_{reject}\}$
to denote the halting states.

Recall that the RMAB minimization problem constraint requires that a reward of at least $R$ in total be obtained from the MDPs in each time step.  Thus in a given time step, if one MDP receives a negative reward, another MDP must receive a positive reward to compensate.

\begin{algorithm}[h] 
    \caption*{Transition Phase of Turing Machine Reduction}
    \label{alg:clean_phase}
    \begin{algorithmic}[1]
    \IF{$j=0 \land k = 0$}
        \IF[Remove current head]{$\text{current}_i$} 
            \STATE $\text{current}_i \gets \texttt{False}$
            \STATE $\text{TMst}_i \gets 0$
        \ELSIF[Add next head]{$\text{TMst}_i \neq 0$}
            \STATE $\text{current}_i \gets \texttt{True}$
            \STATE $\text{TMst}_i, \text{symbol}_i, \text{next}_i \gets \delta(\text{TMst}_i, \text{symbol}_i)$
            \IF[Halting state]{$\text{TMst}_i \in H$} 
                \STATE Enter trapping state \COMMENT{Pay cost $\geq 2\alpha^2$}
            \ENDIF
        \ENDIF
        \STATE Receive reward $R$
    \ENDIF
    \end{algorithmic}
\end{algorithm}

In the copy phase, we copy the state of the Turing Machine to the MDP for the next cell that will contain the head.
The MDP for the current cell receives $R$ reward unless the indices align with the next head and the correct state of the Turing Machine.
Therefore, the correct next cell must be selected and the correct state copied.
Unfortunately, other cells or other states could be selected and copied so we need a validation phase.

\begin{algorithm}[h]
    \caption*{Copy Phase of Turing Machine Reduction}
    \label{alg:copy_phase}
    \begin{algorithmic}[1]
    \IF{$k \in \{1,\ldots,|Q|\}$}
        \IF{$\text{current}_i \land \text{not} (k = \text{TMst}_i \land \text{next}_i=j)$}
            \STATE Receive reward $R$
            \STATE $\triangleright$ Receive reward except when head should be copied
        \ELSIF{not $\text{current}_i \land (\text{played} \land i=j)$}
            \STATE $\text{TMst}_i \gets k$ \COMMENT{Copy TMstate}
            \STATE Receive reward $R$            
        \ENDIF 
    \ENDIF
    \end{algorithmic}
\end{algorithm}

In the validation phase, we ensure that the head and state information was only copied to the correct cell MDP.
If it is copied to a cell MDP,
then that MDP receives $-R$ reward
on the indices $j$ and $k$ corresponding to its index and stored Turing Machine state.
The only way to still satisfy the constraint is if the policy receives $2R$ reward during that time step, which only happens if the indices $j$ and $k$ align with the correct next head and Turing Machine state.
Therefore, the only way to satisfy the reward constraint in every step
of the copy and validation phases is if only the correct head and Turing Machine state are copied.

Therefore, we have reduced the problem of determining whether the Turing Machine $M$ halts on a given input $x$ to the problem of determining an approximately optimal policy for a given RMAB minimization instance.

\begin{algorithm}[h]
    \caption*{Validation Phase of Turing Machine Reduction}
    \label{alg:valid_phase}
    \begin{algorithmic}[1]
    \IF{$k \in \{|Q|+1, \ldots, 2|Q|\}$}
        \IF{not $\text{current}_i \land (i=j \land k= Q+\text{TMst}_i)$}
            \STATE Receive reward $-R$ \COMMENT{Penalize wrong cell}
        \ELSIF{$\text{current}_i \land (k = Q+\text{TMst}_i \land \text{next}_i=j)$}
            \STATE Receive reward $2R$ \COMMENT{Help correct cell}
        \ELSIF{$\text{current}_i$}
            \STATE Receive reward $R$
        \ENDIF       
    \ENDIF
    \end{algorithmic}
\end{algorithm}

\end{proof}

\section{Proof of Corollaries}

\indexability*

\begin{proof}[Proof of Corollary \ref{coro:indexability}]
We first observe that 
$V_{\max}(s,a,\lambda) = \lambda \cdot V_{\min}(s,a,1/\lambda)$.
Now suppose indexability holds for the decoupled maximization problem.
Consider a state $s \in \mathcal{S}$ and real numbers $0 < \lambda' \leq \lambda$.
By the indexability assumption, we have
\begin{align}
    V_{\max}(s,0,1/\lambda) &> V_{\max}(s,1,1/\lambda) \nonumber \\
    \implies
    V_{\max}(s,0,1/\lambda') &> V_{\max}(s,1,1/\lambda')
\end{align}
since $1/\lambda' \geq 1/\lambda$.
By the observation above, this is equivalent to
\begin{align}
    V_{\min}(s,0,\lambda) &> V_{\min}(s,1,\lambda) \nonumber \\
    \implies
    V_{\min}(s,0,\lambda') &> V_{\min}(s,1,\lambda').
\end{align}
By a similar argument, we can show that if indexability holds
for the decoupled minimization problem then it also holds
for the decoupled maximization problem.
The statement follows.
\end{proof}

\index*

\begin{proof}[Proof of Corollary \ref{coro:index}]
Let $\lambda' \leq \lambda^+ \leq \lambda''$.
By indexability and the definition of the Whittle index,
we have 
\begin{align*}
  V_{\max}(s,0,\lambda') &\leq V_{\max}(s,1,\lambda') \\
  V_{\max}(s,0,\lambda'') &> V_{\max}(s,1,\lambda'').
\end{align*}
By the observation relating the decoupled maximization
and minimization problems, this is equivalent to
\begin{align*}
  V_{\min}(s,0,1/\lambda') &\leq V_{\min}(s,1,1/\lambda') \\
  V_{\min}(s,0,1/\lambda'') &> V_{\min}(s,1,1/\lambda'').
\end{align*}
Since the inequality holds for all $1/\lambda' \geq 1/\lambda^+ \geq 1/\lambda''$,
we have that $1/\lambda^+ = \inf \{\lambda: V_{\min}(s,0,\lambda) > V_{\min}(s,1,\lambda)\}$.
By a similar argument, we can show that if the minimization
index is $\lambda^-$ then the Whittle index is $1/\lambda^-$.
The statement follows.
\end{proof}

\section{Data Description}

\subsection*{National Inpatient Sample}

The first data set consists of real hospital visits from the National Inpatient Sample \cite{hcup2018national}.
Each action corresponds to selecting a patient for elective care.
We build the problem by sampling $n$ anonymized patients from the hospital data.
The cost of a patient's elective care is the total amount they were charged for a hospital visit.
The reward of a patient's elective care is how much they needed care as measured by the severity of their condition \cite{averill2003all}.
In order to make the reward function stochastic, we draw the severity of a patient's condition from a distribution over similar patients.
For each patient, we compute the probability similar patients are admitted for emergency care.
For each week a patient does not receive elective care, we compound the probability they need emergency care.
With this probability, the patient needs emergency care and the total reward from the week is \textit{decreased} by the severity of their condition.

Since the problem is based on real data and does not have an analytical form, the Whittle indices cannot be computed in closed form.
Instead, we approximate the Whittle indices using Q-learning \cite{robledo2022qwi}.
Since a tabular approach is infeasible given the number of patients in the data set,
we train a shallow neural network to learn the Q-value.
In our notation, the Q-value is $V_{\min}(\cdot)$, of each patient and action.
In particular, we train the neural network $f_\theta$ with parameters $\theta$ to minimize the loss
$$\mathcal{L}(\theta) = (f_\theta(s,a) - (r(s,a) - c(a) + \beta \max_{a'} f_\theta(s',a') ))^2$$
over instances from simulations of states $s$, selected actions $a$, stochastic next states $s'$, and next actions $a'$.
Then we approximate the Whittle indices by the value of $\lambda$ that would set $f_\theta(s,0) = f_\theta(s,1) + \lambda c(1)$.
We make all of our code and model weights available in the supplementary material.

Our use of the real data set---the National Inpatient Sample---is governed by a data sharing agreement with the Healthcare Cost and Utilization Project \cite{hcup2018national} which prevents us from sharing the data.
However, researchers may apply to access the data by following the process and training outlined on the HCUP website.

\subsection*{Hidden Two-State MDP}

The second data set consists of synthetic data from a well-studied two-state hidden MDP problem \cite{le2008multi,liu2010indexability}.
We choose this synthetic data set because we can compute the Whittle index in closed form and compare the algorithms when the rewards are well-spread.
There are are $n$ MDPs, each in state $0$ or state $1$.
The $i$th MDP transitions from state $x$ to $y$ with probability 
$p^{xy}_i \in (0,1)$ for $x,y \in \{0,1\}$.
The state of an MDP is hidden unless we select it;
if we select an MDP \textit{and} the MDP is in state $1$,
we receive positive reward $r_i \in \mathbb{R}_{>0}$.
The challenge is that we want a strategy to be \textit{exploiting} MDPs that are likely in the 1 state while \textit{exploring} new MDPs to learn if they are in the 1 state.
We consider the problem with unit costs because we are not aware of results for the Whittle index in the more general non-negative cost setting.
We make the first group of reliable MDPs have reward $r_i=1$ and long-term expected reward $1$.
We make the second group of unreliable MDPs have reward $r_i\approx 10n^2$ and long-term expected reward $2$.
We make all our code and data available in the supplementary material.

\section{Maximization Overview}
It is computationally hard to solve the maximization problem exactly in general so prior work considers an approximate version of the problem \cite{whittle1988restless}.

\begin{problem}[Approximate Maximization]
Consider an approximation factor $\alpha \geq 1$.
A policy $\pi$ is an $\alpha$-approximation to the maximization
problem if the expected discounted reward
(Equation \ref{eq:max_objective})
is within a factor of $\alpha$ of the optimal policy $\pi^+$
and the budget constraint is satisfied.
\end{problem}

Since evaluating the optimal strategy for the maximization problem
is PSPACE-hard \cite{papadimitriou1994complexity},
the classical approach is to make a series of relaxations
to get a decoupled problem for each MDP.
When we consider a single MDP $i$, we typically drop the subscript
for notational brevity.

\begin{problem}[Decoupled Maximization \cite{whittle1988restless}]
Consider a particular MDP in the RMAB instance.
Fix $\lambda \geq 0$.
The decoupled maximization problem for a particular MDP 
with parameter $\lambda$ is to find the policy
\begin{align}
    \arg \max_{\pi \colon \mathcal{S} \rightarrow \mathcal{A}}
    \E_{\{a^{(t)}, s^{(t)}\}_{t=1}^\infty} \left[ \sum_{t=1}^\infty \beta^{t-1}
    (r(s^{(t)}, a^{(t)}) - \lambda c(a^{(t)}))
    \right].
\end{align}
\end{problem}

Before we define indexability, we will introduce notation
to describe the expected value of taking an action
in a state and then following the optimal policy.
Let
\begin{align*}
    V_{\max}(s,a,\lambda) &= \E
    \left[ r(s, a) - \lambda c(a) \right] \\
    &+ \E_{s' \sim \tau(s,a)}
    \left[\beta \max_{a'} V_{\max}(s', a', \lambda) \right].
\end{align*}

\begin{definition}[Indexability for Maximization]
An MDP is indexable if for all states $s \in \mathcal{S}$
and real numbers $\lambda' \geq \lambda$,
\begin{align}
    V_{\max}(s,0,\lambda) &> V_{\max}(s,1,\lambda) \nonumber \\
    \implies
    V_{\max}(s,0,\lambda') &> V_{\max}(s,1,\lambda').
\end{align}
In words, if it is optimal to take the passive action under
a subsidy $\lambda$ in state $s$, 
then it must also be optimal to take the passive
action under a larger subsidy $\lambda'$ in state $s$.
\end{definition}

\end{document}